\documentclass{scrartcl}
\usepackage[numbers]{natbib}
\usepackage{fontenc}
\usepackage{shadethm}
\usepackage{amsthm}
\usepackage{float}
\usepackage{bm}
\usepackage{multicol,lipsum}
\usepackage{mathtools}
\usepackage{graphicx}
\usepackage{subfig}
\usepackage{stmaryrd}
\usepackage{mathrsfs}
\usepackage{amsmath}
\usepackage{amssymb}
\usepackage{wasysym}
\usepackage{sectsty}
\usepackage{stackengine}
\usepackage[hyperfootnotes=false]{hyperref}
\usepackage[a4paper, left=2.7cm, right=2.7cm, top=2.5cm]{geometry}
\usepackage{chngcntr}
\counterwithin*{equation}{section}

\usepackage{cleveref}
\DeclareMathAlphabet{\mathpzc}{OT1}{pzc}{m}{it}

\sectionfont{\normalsize\centering}
\subsectionfont{\footnotesize\centering}

\theoremstyle{plain}
\newtheorem{thm}{Theorem}[section] % reset theorem numbering for each chapter

\theoremstyle{definition}
\newtheorem{lem}[thm]{Lemma}

\newtheorem{rem}[thm]{Remark}
\newtheorem{cor}[thm]{Corollary}
\newtheorem{quest}[thm]{Question}

\def\XXint#1#2#3{{\setbox0=\hbox{$#1{#2#3}{\int}$ }
		\vcenter{\hbox{$#2#3$ }}\kern-.6\wd0}}

\usepackage[utf8]{inputenc}
\usepackage[german,english,russian]{babel}

\newcounter{MPequ}

\newcounter{AppA}
\newenvironment{AppA}
{\stepcounter{AppA}%
	\addtocounter{equation}{0}%
	\equation}
{\endequation}
\newcounter{AppB}
\newenvironment{AppB}
{\stepcounter{AppB}%
	\addtocounter{equation}{0}%
	\equation}
{\endequation}
\newcounter{AppC}

\newcounter{AppD}

\newcounter{AppE}

\begin{document}\selectlanguage{english}
\begin{center}
\normalsize \textbf{\textsf{Influence of coil geometry and coil-plasma distance on the magnetic field approximation error}}
\end{center}
\begin{center}
 Wadim Gerner\footnote{\textit{E-mail address:} \href{mailto:wadim.gerner@edu.unige.it}{wadim.gerner@edu.unige.it}}
\end{center}
\begin{center}
{\footnotesize MaLGa Center, Department of Mathematics, Department of Excellence 2023-2027, University of Genoa, Via Dodecaneso 35, 16146 Genova, Italy}
\end{center}
{\small \textbf{Abstract:} 
We investigate analytically two questions:

1) How does the coil geometry influence the effect of electric current noise on the induced magnetic field?

2) How does the coil-plasma distance influence our ability to control the pointwise magnetic field error in terms of the average magnetic field error?

Regarding (1), we argue that the main geometric quantities of interest are the notion of reach and the volume of the region enclosed by the coils. Our main finding is a quantitative formula which shows that the larger the reach and the smaller the volume of the region enclosed by the coils, the smaller is the influence of the electric current uncertainty on the magnetic fields.

Regarding (2), we show that the pointwise magnetic field error can be controlled (modulo an explicit constant) by the average square-magnetic field-error times $(\text{coil-plasma distance})^{-\frac{3}{2}}$.
\newline
\newline
{\small \textit{Keywords}: Coil optimisation, Stellarator, Plasma-coil separation, Coil geometry, Magnetic field approximation}
\begin{multicols}{2}
\section{Introduction}
\label{Intro}
In contrast to the tokamak approach, stellarator devices aim to produce the plasma confining magnetic field by means of complex 3-d coil structures without relying on strong plasma currents, see \cite{Spi58} for pioneering work on the stellarator concept and \cite{Hel14},\cite{IGPW24} for more recent introductions to the stellarator approach. One of the advantages of the stellarator approach is that once it is made to work, the stellarator device may be operated in a steady state, while the main disadvantage is the highly complex coil design \cite{Hel14}, see also \cite{Xu16} for a comparison of the tokamak and stellarator approaches.

In the context of stellarator coil designs there are several approaches to model the coils, cf. \cite[Section 13]{IGPW24} for a concise overview,
\begin{enumerate}
	\item The filament approach, in which the coils are modelled as infinitely thin closed loops,
	\item The massive 3-d coil approach, where each coil is modelled as a solid torus occupying a positive finite volume in $3$-space,
	\item The coil winding surface (CWS) model \cite{M87} in which the coils are modelled as a single toroidal surface along which the coils may be thought to wind around.
\end{enumerate}
In the present work we focus entirely on the CWS model and intend to answer how geometric properties of the CWS itself as well as the coil-plasma distance may influence our ability to approximate well a desired target magnetic field.

We start by pointing out that there are two main approaches to the plasma-coil optimisation problem:
\begin{enumerate}
	\item Two-Stage optimisation approach: In this approach one first optimises the plasma shape and identifies an appropriate plasma confining magnetic field, which can be achieved by plasma equilibrium codes such as DESC \cite{DK20}. In a second step a suitable coil winding surface and current distribution is thought such that the magnetic field induced by said current distribution approximates well in some appropriate sense the target magnetic field obtained from the first step, \cite{L17},\cite{PLBD18},\cite{PRS22}.
	\item Single-Stage optimisation approach: In this approach one seeks simultaneously a plasma shape, plasma equilibrium field and coil winding surface shape. This is advantages because one can strike a balance between the desired plasma properties as well as the coil complexity which guarantees a good target field approximation. While in contrast in the two step optimisation it might be that the plasma is well-optimised, but the required coil structures turn out to be too complex from an engineering point of view to achieve good approximations, \cite{HHPH21},\cite{JGLRW23}.
\end{enumerate}
The focus of the present investigation will be on the two stage procedure. In particular we assume that we already identified some plasma region $P\subset\mathbb{R}^3$ and a target magnetic field $B_T$ on $P$ which we assume to be div-free and curl-free. In practice one seeks a plasma equilibrium fields
\begin{gather}
	\nonumber
	B\times \operatorname{curl}(B)=\nabla p
\end{gather}
where $B$ is the total magnetic field and $p$ is the plasma pressure. Our plasma domain $P$ can then be thought of consisting of the level sets of the pressure function $p$, or more generally of any suitable flux function $\psi$. The plasma current $J$ is according to Maxwell's equations simply given by $J=\operatorname{curl}(B)$ and so in turn the magnetic field induced by the plasma current can be obtained by the Biot-Savart law
\begin{gather}
	\nonumber
	B_{\operatorname{plasma}}=\operatorname{BS}(J)=\frac{\mu_0}{4\pi}\int_PJ(y)\times \frac{x-y}{|x-y|^3}d^3y.
\end{gather}
The corresponding target field which needs to be produced by the coil currents is then $B_T=B-B_{\operatorname{plasma}}$ which is div- and curl-free.

In order to obtain coil winding surfaces one often tries to consider surfaces which are conformal to the outer most plasma surface $\partial P$, \cite{KLM24}. More precisely one looks for surfaces which are of the form $\Psi_r(x):=x+r\mathcal{N}(x)$ where $x\in \partial P$, $\mathcal{N}(x)$ is the outward unit normal at $x$ and $0<r$ is a parameter which corresponds to the distance of the conformal surface $\Sigma_r:=\{\Psi_r(x)\mid x\in \partial P\}$ to the last closed flux surface, i.e. $\operatorname{dist}(\Sigma_r,\partial P)=r$. It is standard \cite[Product neighbourhood theorem]{M65} that for small enough $r$ this provides a diffeomorphism between $\Sigma_r$ and $\partial P$, while for large enough $r$ overlaps may occur.

The largest number $\rho>0$ such that for all $0<r<\rho$, $\Psi_r$ is a diffeomorphism is known as the reach of $\partial P$. This concept dates back to Federer \cite[Section 4]{Fed59} and relationships to other geometric properties are discussed in detail in \cite[Section 2]{Dal18}. In particular, \cite[Remark 2.8]{Dal18} shows that $|\kappa_i|\leq \frac{1}{\rho}$ where $\kappa_i$ denote the principal curvatures of the surface $\partial P$, $i=1,2$.

As we shall see in a moment the reach of the coil winding surface is its main geometric feature which determines how strong current density variations affect the magnetic field.

Before we discuss this result in the next section, let us recall here that for a fixed CWS $\Sigma$ one is looking for a (div-free) current distribution $j$ on $\Sigma$ which approximates well a desired target field $B_T$. The magnetic field induced by $j$ is given by the Biot-Savart law as
\begin{gather}
	\nonumber
	\operatorname{BS}_{\Sigma}(j)(x)=\frac{\mu_0}{4\pi}\int_{\Sigma}j(y)\times \frac{x-y}{|x-y|^3}d\sigma(y).
\end{gather}
Simply minimising $\|\operatorname{BS}_{\Sigma}(j)-B_T\|^2_{L^2(P)}$ leads to an ill-posed optimisation problem so that a Tikhonov regularisation, \cite{L17}, can be used to regularise the problem by penalising the average current strength $\|j\|^2_{L^2(\Sigma)}$ which leads to the minimisation problem
\begin{gather}
	\nonumber
	j_{\lambda}=\underset{j}{\operatorname{argmin}}\left\{\|\operatorname{BS}_{\Sigma}(j)-B_T\|^2_{L^2(P)}+\lambda \|j\|^2_{L^2(\Sigma)}\right\}
\end{gather}
where $\lambda>0$ is the penalisation\slash regularisation parameter. In practice, \cite{L17}, and in particular in the REGCOIL procedure it is customary to replace the expression $\|\operatorname{BS}_{\Sigma}(j)-B_T\|^2_{L^2(P)}$ by $\|\operatorname{BS}_{\Sigma}(j)\cdot \mathcal{N}-B_T\cdot \mathcal{N}\|^2_{L^2(\partial P)}$. It has been shown in \cite[Lemma 11]{PRS22} that, upon fixing the harmonic part of $j$ according to the toroidal circulation of the target field $B_T$, the condition $\|\operatorname{BS}_{\Sigma}(j)\cdot \mathcal{N}-B_T\cdot \mathcal{N}\|^2_{L^2(\partial P)}\ll 1$ implies $\|\operatorname{BS}_{\Sigma}(j)-B_T\|^2_{L^2(P)}\ll 1$. The converse is however not true and we refer the reader to \cite[Theorem A.1]{G25BiotSavartImageKernelArXiv} for a more precise relationship between $\|\operatorname{BS}_{\Sigma}(j)-B_T\|^2_{L^2(P)}$ and certain boundary norms of $\operatorname{BS}_{\Sigma}(j)\cdot \mathcal{N}-B_T\cdot \mathcal{N}$. We also emphasise that it has been shown recently \cite[Section 4]{G24} that if we work with the expression $\|\operatorname{BS}_{\Sigma}(j)-B_T\|^2_{L^2(P)}$ instead of the $L^2(\partial P)$-norm of their normal parts, then we have the convergence $\|\operatorname{BS}_{\Sigma}(j_{\lambda})-B_T\|_{L^2(P)}\rightarrow 0$ as $\lambda\rightarrow 0$. So that at least theoretically every target magnetic field can be approximated arbitrarily well. Of course, in practice the corresponding currents $j_{\lambda}$ may be so complex that they are not realisable due to engineering constraints. We refer the reader here to \cite{KLM24} for a discussion of the approximation quality from a more practical perspective, taking into account engineering constraints such as the coil-coil distance.

We emphasise here that procedures such as REGCOIL, \cite{L17}, provide approximations of the target field in some averaged sense $\|\operatorname{BS}_{\Sigma}(j_{\lambda})-B_T\|_{L^2(P)}$ (or at best in some fractional Sobolev space $H^s(P)$ with $0<s\leq\frac{1}{2}$, cf. \cite[Section A.3]{PRS22}). In practice we are interested in approximating our target field well at all points. More precisely, a good average approximation, allows for the possibility of the current induced magnetic field $\operatorname{BS}_{\Sigma}(j_{\lambda})$ to differ greatly from the target field in some small region of $P$, which can heavily affect the plasma behaviour. The quantity in which we are more interested in is actually
\begin{gather}
	\nonumber
	\max_{x\in P}|\operatorname{BS}_{\Sigma}(j)(x)-B_T(x)|.
\end{gather}
This quantity may be large even if the $L^2$-distance is small. The main observation here is however that the target field, as well as, $\operatorname{BS}_{\Sigma}(j_{\lambda})$ satisfy the vector Laplace's equation $\Delta B_T=0=\Delta \operatorname{BS}_{\Sigma}(j_{\lambda})$ so that the components of $\operatorname{BS}_{\Sigma}(j)-B_T$ are harmonic functions. One can then exploit the mean-value property \cite[Section 2.2.2 Theorem 2]{Evans10} satisfied by such functions to control the pointwise error by means of the $L^2$-error, albeit this necessitates to keep the $L^2$-error small on a larger region than the region within which we desire to obtain a good pointwise approximation. We shall see that this approach yields the best possible constant in the corresponding error-estimate, see \Cref{AppAR2} for a more precise explanation. An alternative approach, which yields qualitatively the same result, can be obtained by means of elliptic estimates, cf. \cite[Section 6.3]{Evans10}.

The first question we want to consider here is therefore the following:
\begin{quest}
	\label{S1Q1}
	\textit{How can we control the pointwise error by means of the $L^2$-error and how does the geometry of the CWS and the coil-plasma distance influence our ability to do so?}
\end{quest}
In addition to that, even if the target magnetic field can be reproduced exactly by some current distribution $j$ on the CWS $\Sigma$, in practice, we will not be able to reproduce a current distribution exactly. In addition, currents produced in real life applications are bound by additional material laws, which are usually not incorporated in the standard optimisation procedure, so that the physical currents $j_{\operatorname{ph}}$ will differ from the optimal current $j_{\lambda}$ obtained from our minimisation procedure, $\|j_{\operatorname{ph}}-j_{\lambda}\|_{L^2(\Sigma)}>0$. The question then becomes how the current error affects the magnetic field, $\|\operatorname{BS}_{\Sigma}(j_{\operatorname{ph}})-\operatorname{BS}_{\Sigma}(j_{\lambda})\|_{L^2(P)}$, and in particular how the coil geometry may strengthen or weaken the influence of the current error on the induced magnetic fields.
\begin{quest}
	\label{S1Q2}
	\textit{How does the CWS geometry affect the influence of the current error on the magnetic field error?}
\end{quest}
In the following two sections we provide answers to these two questions.
\section{$L^2$-error vs. pointwise error}
In this section we discuss how one can control the pointwise magnetic field error by means of the $L^2$-error. As alluded to in the previous section this will be possible only upon approximating the target field in $L^2$ on a larger region. The following is a slightly informal statement of our main finding regarding \Cref{S1Q1}. For a mathematically rigorous statement and a proof see \Cref{AppA}.
\begin{thm}
	\label{S2T1}
	Let $\Sigma$ be our CWS, $P$ be the plasma region and $\Omega$ be the finite region enclosed by $\Sigma$, $\partial \Omega=\Sigma$. Let further $P\subset U\subseteq \Omega$, then for every current $j$ on $\Sigma$ we have the following inequality
	\begin{gather}
		\nonumber
		\max_{x\in P}|\operatorname{BS}_{\Sigma}(j)(x)-B_T(x)|\leq
		\\
		\label{S2E1}
		\sqrt{\frac{3}{4\pi(\operatorname{dist}(\partial P,\partial U))^3}}\|\operatorname{BS}_{\Sigma}(j)-B_T\|_{L^2(U)}.
	\end{gather}
\end{thm}
\Cref{S2T1} tells us that if we approximate the target field on some larger region $U$ in some averaged sense, then we can also control the pointwise error on the smaller plasma region $P$ up to a factor which takes into account by how much the region $U$ is larger than the plasma region $P$. In essence, the larger the region $U$ on which we have a good $L^2$-error control, the better we will be able to control the pointwise error within the plasma region. The factor $\sqrt{\frac{3}{4\pi}}$ is optimal in this type of inequality, see \Cref{AppA}, with $\sqrt{\frac{3}{4\pi}}\approx 0.488\dots<\frac{1}{2}$.

We also notice that since $U\subseteq \Omega$, we find $\operatorname{dist}(\partial P,\partial U)\leq \operatorname{dist}(\partial P,\partial \Omega)=\operatorname{dist}(\partial P,\Sigma)$ so that the best case scenario we can aim for is that we are able to control the $L^2$-error on the whole finite region $\Omega$ enclosed by the CWS. In that case we obtain
\begin{cor}
	\label{S2C2}
	Let $\Sigma$ be our CWS, $P$ be our plasma region and $\Omega$ be the finite region enclosed by $\Sigma$, then
	\begin{gather}
		\nonumber
		\max_{x\in P}|\operatorname{BS}_{\Sigma}(j)(x)-B_T(x)|\leq
		\\
		\label{S2E2}
		\sqrt{\frac{3}{4\pi(\operatorname{dist}(\partial P,\Sigma))^3}}\|\operatorname{BS}_{\Sigma}(j)-B_T\|_{L^2(\Omega)}.
	\end{gather}
\end{cor}
As described in \Cref{Intro} a CWS can be obtained from the last closed flux surface $\partial P$ by means of a diffeomorphism $\Psi_r(x)=x+r\mathcal{N}(x)$ and the largest $r>0$ for which this gives a non-self intersecting surface is the reach of $\partial P$. This shows that in this approach the plasma geometry provides us with an upper bound on how tightly the $L^2$-error can control the pointwise error.

To make this precise we write $\Sigma_P$ for the CWS obtained from $\partial P$ through $\Psi_{\rho}(x)=x+\rho\mathcal{N}(x)$ where $\rho$ denotes the reach of $\partial P$.
\begin{cor}
	\label{S2C3}
	Let $P$ be our plasma region, $\rho$ be the reach of $\partial P$ and $\Omega$ be the finite region enclosed by $\Sigma_P$. Then
	\begin{gather}
		\nonumber
		\max_{x\in P}|\operatorname{BS}_{\Sigma}(j)(x)-B_T(x)|
		\\
		\label{S2E3}
		\leq \sqrt{\frac{3}{4\pi\rho^3}}\|\operatorname{BS}_{\Sigma}(j)-B_T\|_{L^2(\Omega)}.
	\end{gather}
\end{cor}
We emphasise that the observations made in \Cref{S2T1} and \Cref{S2C2} remain valid even if the CWS is not conformal to $\partial P$ with the caveat that in this case the plasma-coil distance cannot be related to the geometry of $\partial P$ in general.
\section{CWS geometry and current error propagation}
In this section we want to discuss how the geometry of our CWS influences the effect of a current density error $\|j-j_*\|_{L^2(\Sigma)}$ on the magnetic field error $\|\operatorname{BS}_{\Sigma}(j)-\operatorname{BS}_{\Sigma}(j_*)\|_{L^2(\Omega)}$.

Our main finding is the following, see \Cref{AppB} for a mathematically more rigorous statement and a proof.
\begin{thm}
	\label{S3T1}
	Let $\Sigma$ be our CWS, $\Omega$ be the finite region enclosed by $\Sigma$ and $\rho_{\Sigma}$ be the reach of $\Sigma$. Then for any two current distributions $j$, $j_{*}$ on $\Sigma$ we have
	\begin{gather}
		\nonumber
		\|\operatorname{BS}_{\Sigma}(j)-\operatorname{BS}_{\Sigma}(j_{*})\|_{L^2(\Omega)}
		\\
\label{S3E1}
\leq\frac{2\mu_0\sqrt[3]{\operatorname{vol}(\Omega)}}{\sqrt{\rho_{\Sigma}}}\|j-j_{*}\|_{L^2(\Sigma)}
	\end{gather}
	where $\mu_0=1.256 637 061 27 \cdot 10^{-6}\frac{N}{A^2}$ is the vacuum magnetic permeability.
\end{thm}
The factor $2$ is not optimal and in fact our methods yield a slightly better constant, cf. \Cref{AppB}. But there is no reason to believe that the methods we employ provide the best possible constant so that we decided to round up the factor to the next largest integer number.

\Cref{S3T1} tells us how the geometry of the CWS and current errors affect the average magnetic field error within the whole finite region $\Omega$ enclosed by the CWS. However, we mainly care about the plasma region, where according to \Cref{S2C2} we can control the pointwise magnetic field error in terms of the $L^2$-error on the full domain $\Omega$ and hence in terms of the current density error $\|j-j_*\|_{L^2(\Sigma)}$.

We therefore arrive at the following conclusion
\begin{cor}
	\label{S3C2}
	Let $\Sigma$ be our CWS, $P$ be our plasma region and $\rho_{\Sigma}$ be the reach of $\Sigma$. Then for any two currents $j,j_*$ on $\Sigma$ we have
	\begin{gather}
		\nonumber
		\max_{x\in P}|\operatorname{BS}_{\Sigma}(j)(x)-\operatorname{BS}_{\Sigma}(j_*)(x)|
		\\
		\label{S3E2}
		\leq \sqrt{\frac{3\mu^2_0(\operatorname{vol}(\Omega))^{\frac{2}{3}}}{\pi\rho_{\Sigma}(\operatorname{dist}(\partial P,\Sigma))^3}}\|j-j_*\|_{L^2(\Sigma)}.
	\end{gather}
\end{cor}
\section{Discussion and Summary}
We investigated how the average ($L^2$-)magnetic field approximation which (in one way or another) is targeted in established minimisation procedures such as REGCOIL can control the pointwise magnetic field error, which is of more practical relevance than the $L^2$-error itself.

Our main finding is that as long as we approximate the target field on a larger region $U$ in an averaged sense, we obtain a pointwise control on the (smaller) plasma region $P$. The conversion factor between these two distinct types of errors is proportional to $\operatorname{dist}(\partial P,\partial U)^{-\frac{3}{2}}$.

This in particular tells us that if we want to guarantee a good pointwise approximation of the target field we should aim to approximate the target field in $L^2$ on a strictly larger region than the plasma region of interest.

We have further argued that if the CWS is taken to be conformal to the last closed flux surface, then the maximal distance between the CWS and the last flux surface is determined by the geometry of the last flux surface, more specifically by its reach. This tells us that the reach of the last closed flux surface is, in the context of the pointwise magnetic field error control, the main geometric quantity of interest and should therefore be targeted in plasma optimisation procedures.

Further, we discussed how the coil-plasma distance as well as the geometry of the CWS influence the propagation of the (average) induced current error to the magnetic field error.

To complete our discussion let us consider once again the situation in which we take the CWS to be conformal to the last closed flux surface. In that case $\Sigma=\Psi_r(\partial P)$ where $0<r<\rho$ and $\rho$ is the reach of $\partial P$. We notice that while (\ref{S2E2}) tells us that it is advantageous to pick $r$ as large as possible, i.e. $r\approx\rho$, it follows from (\ref{S3E2}) that this might be a bad approach, because by definition of the reach it will imply that there will be distinct points on $\Sigma$ which almost touch each other, i.e. the reach of $\Sigma$ will be small so that the factor on the right hand side of (\ref{S3E2}) becomes large, cf. \Cref{F1}. This suggests that the pointwise magnetic field error can be sensitive to noise with respect to the induced electric current. If we set $r=\alpha \rho$ with $0<\alpha<1$ we see that $\operatorname{dist}(\partial P,\Sigma)=r=\alpha \rho$. In addition, it could be, in the worst case, that two points $x,y\in \partial P$ are moved closer to each other by a distance of $r$ each, which in turn may lead to the situation where $\rho_{\Sigma}=\rho-r=(1-\alpha)\rho$, see \Cref{F1}. In general however, the reach of $\Sigma$ may be strictly larger than that of $\partial P$, see again \Cref{F1}. If we assume here the worst case scenario $\rho_{\Sigma}=(1-\alpha)\rho$, then the prefactor $\frac{1}{\sqrt{\rho_{\Sigma}}\sqrt{\alpha\rho}^3}$ becomes $\frac{1}{\sqrt{1-\alpha}\sqrt{\alpha}^3}\frac{1}{\rho^2}$. One should therefore try to strike a balance between the factor $\frac{1}{\sqrt{1-\alpha}\sqrt{\alpha}^3}$ appearing in (\ref{S3E2}) and $\frac{1}{\sqrt{\alpha}^3}$ appearing in (\ref{S2E2}). In fact, one should also take into account the factor $\sqrt[3]{\operatorname{vol}(\Omega)}^2$. To do this we observe that
\begin{gather}
	\nonumber
	\operatorname{vol}(\Omega)=\operatorname{vol}(P)+V(r)
\end{gather}
where $V(r)>0$ is the additional volume which is enclosed by the CWS in comparison to the volume enclosed by the last closed flux surface. For a conformal CWS surface one has the formula
\begin{gather}
	\label{S4E1}
	V(r)=r|\partial P|-r^2\int_{\partial P}Hd\sigma+\frac{r^3}{3}\int_{\partial P}\kappa d\sigma
\end{gather}
where $H=\frac{\kappa_1+\kappa_2}{2}$ is the mean curvature of $\partial P$ (w.r.t. the outward unit normal), $\kappa$ is the Gauss curvature of $\partial P$ and $|\partial P|$ denotes the area of $\partial P$. By the Gauss-Bonnet theorem \cite[Theorem 9.7]{L18} we have $\int_{\partial P}\kappa d\sigma=2\pi\chi(\partial P)=0$ where $\chi(\partial P)$ is the Euler characteristic of $\partial P$ which equals $0$ since $\partial P$ is a torus in our application. Keeping in mind that $r=\alpha \rho$, we obtain the formula
\begin{gather}
	\label{S4E2}
	\operatorname{vol}(\Omega)=\operatorname{vol}(P)+\rho\left(\alpha|\partial P|-\alpha^2\rho\int_{\partial P}Hd\sigma\right).
\end{gather}
We overall conclude that there are some competing interests. On the one hand one should make the coil-plasma distance as large as possible in order to be able to control the pointwise magnetic field error better in terms of its $L^2$-counter part, cf. \Cref{S2C2}. On the other hand, if the coil plasma distance (for a conformal CWS) becomes too large, the reach of the CWS may become small which in turn results in a higher sensitivity of the pointwise magnetic field error with respect to noise of the electric current. Therefore one should attempt to look for a middle ground which, to some extent, accommodates both properties.
\end{multicols}
\begin{figure}[H]
	\centering
	\begin{subfloat}
		\centering
		\includegraphics[width=0.5\textwidth, keepaspectratio]{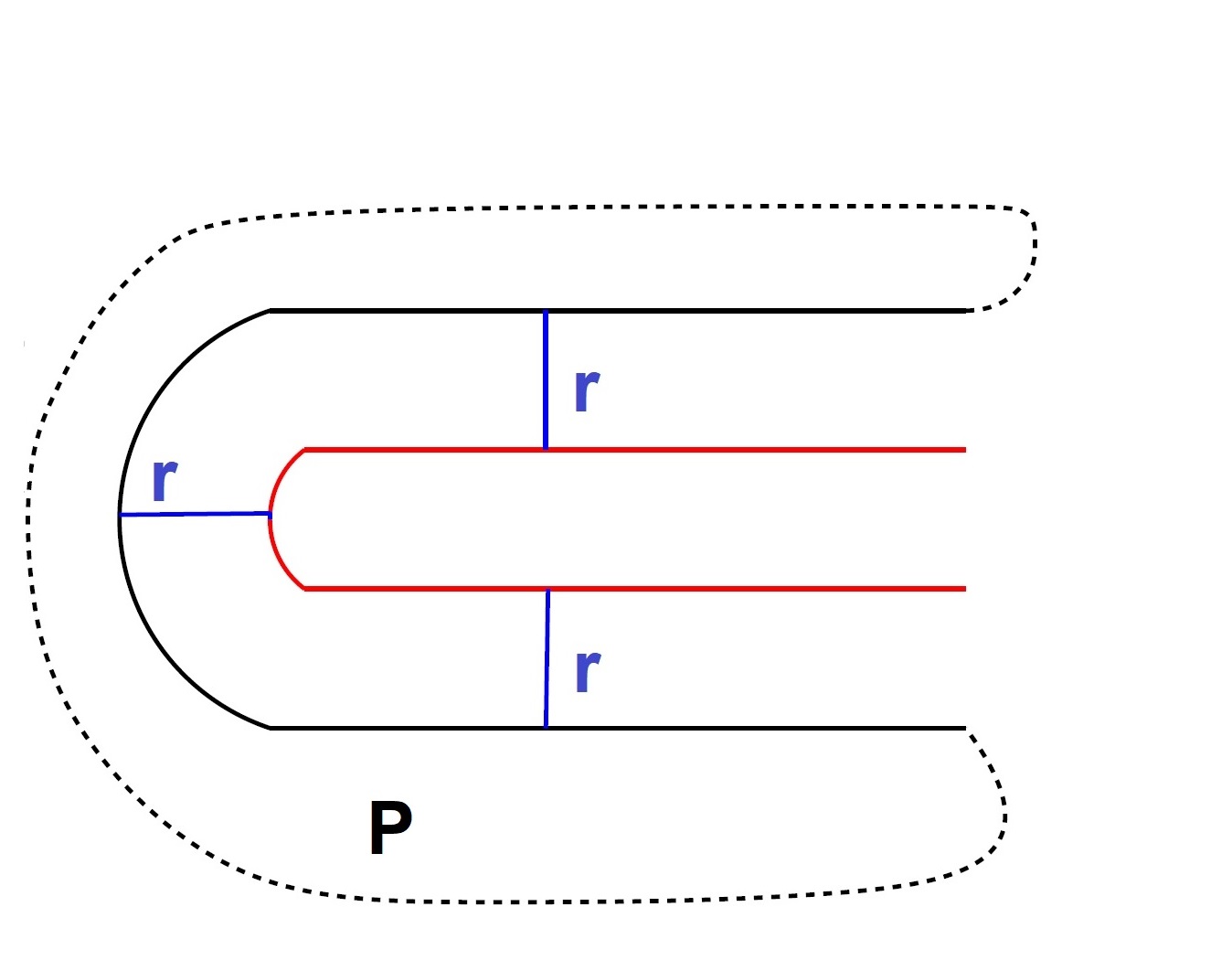}
	\end{subfloat}
	\hspace{1cm}
	\begin{subfloat}
		\centering
		\includegraphics[width=0.35\textwidth, keepaspectratio]{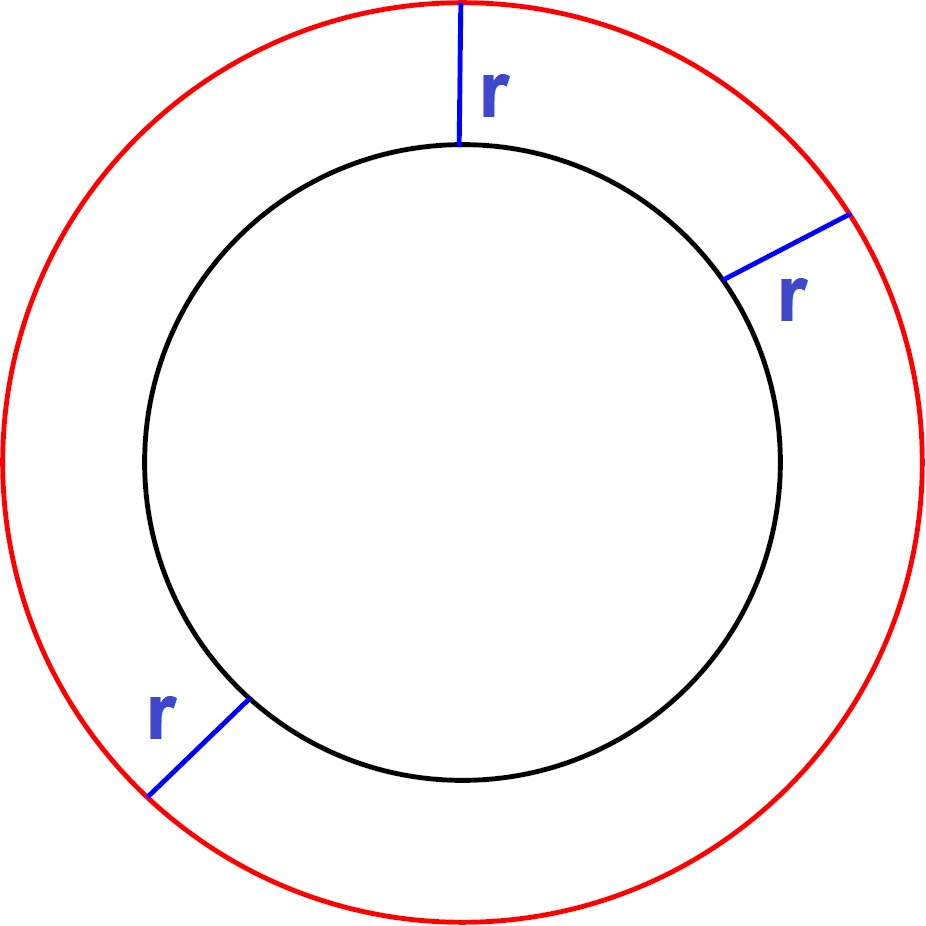}
	\end{subfloat}
	\caption{Left side: The CWS (in red) is obtained from $\partial P$ (in black) by displacing each point on $\partial P$ in outward normal direction. This can bring points closer together and reduce the reach of the CWS in comparison to the reach of $\partial P$.
		\newline
	Right side: It can also happen that all points are moved further apart from each other and the reach of the CWS is larger than the reach of $\partial P$.}
	\label{F1}
\end{figure}
\section*{Data availability}
No new data was created during the course of this work.
\section*{Conflict of interest}
The author declares that he has no conflict of interest.
\section*{Acknowledgements}
The research was supported in part by the MIUR Excellence Department Project awarded to Dipartimento di Matematica, Università di Genova, CUP D33C23001110001.
\appendix
\section{Proof of pointwise error estimates}
\label{AppA}
Here we prove \Cref{S2T1}. The following is the precise mathematical formulation.
\begin{thm}
	\label{AppAT1}
	Let $U\subset\mathbb{R}^3$ be an open set and $P\subset U$ be any subset such that $\operatorname{dist}(P,\partial U)>0$. Then for every $B\in L^2(U,\mathbb{R}^3)$ with $\Delta B=0$ in $U$ in the distributional sense we have the estimate
	\begin{AppA}
		\label{AppAE1}
		|B(x)|\leq \sqrt{\frac{3}{4\pi (\operatorname{dist}(P,\partial U))^3}}\|B\|_{L^2(U)}\text{ for all }x\in P.
	\end{AppA}
\end{thm}
If $P$ is our plasma domain, then of course $\operatorname{dist}(P,\partial U)=\operatorname{dist}(\partial P,\partial U)$ and we recover the form of the inequality in \Cref{S2T1}.
\begin{rem}
	\label{AppAR2}
	The factor $\sqrt{\frac{3}{4\pi(\operatorname{dist}(P,\partial U))^3}}$ is the best possible one can do in the generality stated in \Cref{AppAT1}. To see this we may take $U=B_1(0)$, $P=\{0\}$ and $B=(1,0,0)$, then $1=|B(0)|$, $\operatorname{dist}(P,\partial U)=1$ and $\|B\|_{L^2(U)}=\sqrt{\operatorname{vol}(B_1(0))}=\sqrt{\frac{4\pi}{3}}$ and thus equality is attained in (\ref{AppAE1}).
\end{rem}
\begin{proof}[Proof of \Cref{AppAT1}]
	By Weyl's lemma \cite[Lemma 2]{Weyl40} $B$ is analytic in $U$ and so we can make sense of $B(x)$ at every point of $P$. We then make use of the mean-value property of harmonic functions, cf. \cite[Section 2.2.2 Theorem 2]{Evans10}. We expand $|B(x)|^2=\sum_{i=1}^3B^2_i(x)$ where $B_i$ denotes the $i$-th component in the standard Euclidean coordinate system. Since $B$ is harmonic so is each component. Then the mean value property tells us that
	\begin{gather}
		\nonumber
		B_i(x)=\frac{3}{4\pi r^3}\int_{B_r(x)}B_i(y)d^3y
	\end{gather}
	for every $r$ with $\overline{B_r(x)}\subset U$. We notice that for every $x\in P$ we have $B_{\operatorname{dist}(P,\partial U)}(x)\subset U$ and hence for every $0<r<\operatorname{dist}(P,\partial U)$ we have $\overline{B_r(x)}\subset U$ for all $x\in P$. We can therefore estimate for any such $r$
	\begin{AppA}
		\label{AppAE2}
		|B_i(x)|\leq\frac{3}{4\pi r^3}\int_{B_r(x)}|B_i(y)|d^3y\leq \sqrt{\frac{3}{4\pi r^3}}\|B_i\|_{L^2(B_r(x))}\leq\sqrt{\frac{3}{4\pi r^3}}\|B_i\|_{L^2(U)}
	\end{AppA}
	where we used the Cauchy-Schwarz inequality. Taking the limit $r\nearrow \operatorname{dist}(\partial U,P)$ in (\ref{AppAE2}) we obtain
	\begin{gather}
		\nonumber
		|B_i(x)|\leq \sqrt{\frac{3}{4\pi(\operatorname{dist}(P,\partial U))^3}}\|B_i\|_{L^2(U)}\text{ for all }x\in P\text{ and all }1\leq i\leq 3.
	\end{gather}
	We can square the above inequality, take the sum over $i=1,..,3$ and finally take the square root to arrive at the desired conclusion.
\end{proof}
\section{Proof of current error propagation}
\label{AppB}
The goal of this section is to prove \Cref{S3T1}. For a given, closed, $C^{1,1}$-surface $\Sigma$ we set $L^2\mathcal{V}(\Sigma):=\{j\in L^2(\Sigma,\mathbb{R}^3)\mid j\cdot \mathcal{N}=0\text{ for }\mathcal{H}^2\text{ a.e. }x\in \Sigma\}$ where $\mathcal{H}^2$ denotes the $2$-dimensional Hausdorff measure. Further, we denote by $L^2\mathcal{V}_0(\Sigma):=\{j\in L^2\mathcal{V}(\Sigma)\mid \int_{\Sigma}j\cdot \nabla fd\sigma=0\text{ for all }f\in C^{\infty}(\mathbb{R}^3\}$, i.e. the space of div-free currents, which are tangent to $\Sigma$. The precise mathematical formulation of \Cref{S3T1} is then the following
\begin{thm}
	\label{AppBT1}
	Let $\Omega\subset\mathbb{R}^3$ be a bounded domain with $C^{1,1}$-boundary and let $\rho(\Sigma)$ denote its reach. Then for every $j\in L^2\mathcal{V}_0(\partial \Omega)$ we have the estimate
	\begin{AppB}
		\label{AppBE1}
		\|\operatorname{BS}_{\partial\Omega}(j)\|_{L^2(\Omega)}\leq \frac{\sqrt{3\pi^{\frac{2}{3}}+6\sqrt{12^{\frac{2}{3}}+\pi^{\frac{4}{3}}}}}{\sqrt[3]{16}}\frac{\sqrt[3]{\operatorname{vol}(\Omega)}}{\sqrt{\rho(\Sigma)}}\|j\|_{L^2(\partial\Omega)}
		\end{AppB}
	where $\operatorname{BS}_{\partial\Omega}(j)(x):=\frac{1}{4\pi}\int_{\partial\Omega}j(y)\times \frac{x-y}{|x-y|^3}d\sigma(y)$.
\end{thm}
For convenience we dropped here the factor $\mu_0$ and consider the case $j_*=0$ in comparison to \Cref{S3T1}. Due to the linearity of the Biot-Savart operator this is not a restriction. We also have $\frac{\sqrt{3\pi^{\frac{2}{3}}+6\sqrt{12^{\frac{2}{3}}+\pi^{\frac{4}{3}}}}}{\sqrt[3]{16}}\approx 1.994\dots<2$ so that \Cref{S3T1} is implied by \Cref{AppBT1}.

The derivation of \Cref{AppBT1} is unfortunately more involved than the proof of \Cref{AppAT1}. We will require first a technical lemma. We recall here the notion of the space $H(\Omega,\operatorname{curl})$ which consists of all square-integrable vector fields on a domain $\Omega$ whose curl is also square integrable.
\begin{lem}
	\label{AppBL2}
	Let $\Omega\subset \mathbb{R}^3$ be a bounded $C^{1,1}$-domain and let $\rho$ denote the reach of its boundary. Then there exists a linear bounded extension operator
	\begin{gather}
		\nonumber
		E:L^2\mathcal{V}_0(\partial\Omega)\rightarrow H(\Omega,\operatorname{curl})
	\end{gather}
	with the following properties
	\begin{enumerate}
		\item $\mathcal{N}\times E(j)=j$ for all $j\in L^2\mathcal{V}_0(\partial\Omega)$,
		\item $\operatorname{div}(E(j))=0$ for all $j\in L^2\mathcal{V}_0(\partial\Omega)$,
		\item For all $0<\epsilon<\infty$ we have
		\begin{gather}
			\nonumber
			\|E(j)\|_{\epsilon}\leq \sqrt[4]{9\epsilon\left(9\epsilon+3+\frac{1}{\epsilon}\right)}\|j\|_{L^2(\partial\Omega)}
		\end{gather}
		where $\|E(j)\|^2_\epsilon:=\frac{\|E(j)\|^2_{L^2(\Omega)}}{\alpha \rho}+\alpha\rho\|\operatorname{curl}(E(j))\|^2_{L^2(\Omega)}$ with $\alpha:=\sqrt{\frac{9\epsilon}{9\epsilon+3+\frac{1}{\epsilon}}}$.
	\end{enumerate}
\end{lem}
\begin{proof}[Proof of \Cref{AppBL2}]
	We start by defining the space $H^1_{\operatorname{T}}(\Omega,\mathbb{R}^3):=\{A\in H^1(\Omega,\mathbb{R}^3)\mid \operatorname{div}(A)=0\text{ and }\mathcal{N}\cdot A=0\}$. For given $r>0$ we define the inner product
	\begin{gather}
		\nonumber
		\langle A,B\rangle_r:=\frac{\langle A,B\rangle_{L^2(\Omega)}}{r}+r\langle \operatorname{curl}(A),\operatorname{curl}(B)\rangle_{L^2(\Omega)}.
	\end{gather}
	It is then standard that $\langle\cdot,\cdot\rangle_r$ turns $H^1_{\operatorname{T}}(\Omega,\mathbb{R}^3)$ into a Hilbert space with a norm which is equivalent to the $H^1$-norm, cf. \cite[Proof of Lemma 5.4]{G24}. We then follow the reasoning of the proof of \cite[Lemma 5.4]{G24} and define for given $j\in L^2\mathcal{V}_0(\partial\Omega)$ the functional $J:H^1_{\operatorname{T}}(\Omega,\mathbb{R}^3)\rightarrow \mathbb{R}$, $\psi\mapsto \int_{\partial\Omega}j\cdot \psi d\sigma$ which gives rise to a linear, continuous functional. According to Riesz representation theorem we find a unique $A_j$ with $\langle A_j,B\rangle_r=J(B)$ for all $B\in H^1_{\operatorname{T}}(\Omega,\mathbb{R}^3)$. We can then set $v_j:=-r\operatorname{curl}(A_j)$ and we verify identically as in \cite[Lemma 5.4]{G24} that $v_j\in H(\Omega,\operatorname{curl})$ satisfies $\mathcal{N}\times v_j=j$ and $\operatorname{curl}(v_j)=\frac{A_j}{r}$ from which deduce
	\begin{AppB}
		\label{AppBE2}
		\|v_j\|^2_r=\|A_j\|^2_r=\int_{\partial\Omega}j\cdot A_jd\sigma\leq \|j\|_{L^2(\partial\Omega)}\|A_j\|_{L^2(\partial\Omega)}
	\end{AppB}
	where we used the Cauchy-Schwarz inequality. We can then use \cite[Lemma 4.1]{G25GaffneyKornArXiv} to deduce that 
	\begin{gather}
		\nonumber
	\|A_j\|^2_{L^2(\partial\Omega)}\leq \epsilon\rho\|\nabla A_j\|^2_{L^2(\Omega)}+\frac{3+\frac{1}{\epsilon}}{\rho}\|A_j\|^2_{L^2(\Omega)}\text{ for all }0<\epsilon<\infty.
\end{gather}
In addition \cite[Theorem 3.1 \& Equation (4.21)]{G25GaffneyKornArXiv} gives us the estimate
\begin{gather}
	\nonumber
	\|\nabla A_j\|^2_{L^2(\Omega)}\leq 9\left(\|\operatorname{curl}(A_j)\|^2_{L^2(\Omega)}+\frac{\|A_j\|^2_{L^2(\Omega)}}{\rho^2}\right)
\end{gather}
where we used that $\operatorname{div}(A_j)=0$ and $A_j\parallel \partial\Omega$.

We combine the last two inequalities to deduce
\begin{gather}
	\nonumber
	\|A_j\|^2_{L^2(\partial\Omega)}\leq 9\epsilon\rho \|\operatorname{curl}(A_j)\|^2_{L^2(\Omega)}+\left(9\epsilon+3+\frac{1}{\epsilon}\right)\frac{\|A_j\|^2_{L^2(\Omega)}}{\rho}.
\end{gather}
We recall that $\|\operatorname{curl}(A_j)\|^2_{L^2(\Omega)}=\frac{\|v_j\|^2_{L^2(\Omega)}}{r^2}$ and $\|A_j\|^2_{L^2(\Omega)}=r^2\|\operatorname{curl}(v_j)\|^2_{L^2(\Omega)}$. We can insert this in the last inequality and utilise (\ref{AppBE2}) to arrive at
\begin{AppB}
	\label{AppB3}
	\|v_j\|^4_r\leq \|j\|^2_{L^2(\partial\Omega)}\left(9\epsilon\frac{\rho}{r^2}\|v_j\|^2_{L^2(\Omega)}+\left(9\epsilon+3+\frac{1}{\epsilon}\right)\frac{r^2}{\rho}\|\operatorname{curl}(v_j)\|^2_{L^2(\Omega)}\right).
\end{AppB}
We now make the choice $r=\alpha \rho$ for some $0<\alpha<\infty$ which gives us the identity
\begin{gather}
	\nonumber
	9\epsilon\frac{\rho}{r^2}\|v_j\|^2_{L^2(\Omega)}+\left(9\epsilon+3+\frac{1}{\epsilon}\right)\frac{r^2}{\rho}\|\operatorname{curl}(v_j)\|^2_{L^2(\Omega)}=\frac{9\epsilon}{\alpha}\frac{\|v_j\|^2_{L^2(\Omega)}}{r}+\left(9\epsilon+3+\frac{1}{\epsilon}\right)\alpha r\|\operatorname{curl}(v_j)\|^2_{L^2(\Omega)}.
\end{gather}
Therefore, if we select $\alpha\in (0,\infty)$ such that $\frac{9\epsilon}{\alpha}=\left(9\epsilon+3+\frac{1}{\epsilon}\right)\alpha$ and we denote by $\beta^2$ this common value, then (\ref{AppB3}) becomes
\begin{AppB}
	\label{AppBE4}
	\|v_j\|^4_r\leq \beta^2 \|j\|^2_{L^2(\partial\Omega)}\|v_j\|^2_r\Leftrightarrow\|v_j\|_{L^2(\Omega)}\leq \beta \|j\|_{L^2(\partial\Omega)}.
\end{AppB}
It is straightforward to verify that the desired equality is achieved precisely when $\alpha=\sqrt{\frac{9\epsilon}{9\epsilon+3+\frac{1}{\epsilon}}}$ and that the common value is given by $\beta^2=\sqrt{9\epsilon}\sqrt{9\epsilon+3+\frac{1}{\epsilon}}$. We recall that $r=\alpha\rho$ and notice that the assignment $j\mapsto v_j$ is linear because the Riesz isomorphism is linear so that we may set $E(j):=v_j$. Consequently the lemma follows.
\end{proof}
\begin{proof}[Proof of \Cref{AppBT1}]
	By a density argument, we may suppose that the current $j$ is of class $W^{\frac{1}{2},2}$. Then according to \Cref{AppBL2} we see that $E(j)\in H(\Omega,\operatorname{curl})$, $\operatorname{div}(E(j))\in L^2(\Omega)$ and $\mathcal{N}\times E(j)\in W^{\frac{1}{2},2}(\partial\Omega)$ so that \cite[Corollary 2.15]{ABDG98} implies $E(j)\in H^1(\Omega,\mathbb{R}^3)$. It then follows from \cite[Proof of Lemma 5.5]{G24} that we have the identity
	\begin{gather}
		\nonumber
		\operatorname{BS}_{\partial\Omega}(j)(x)=\operatorname{BS}_{\Omega}(\operatorname{curl}(E(j)))(x)+\frac{\nabla_x\int_{\Omega}E(j)(y)\cdot \frac{x-y}{|x-y|^3}d^3y}{4\pi}-E(j)(x)
	\end{gather}
	where $\operatorname{BS}_{\Omega}(\operatorname{curl}(E(j)))(x)=\frac{\int_{\Omega}\operatorname{curl}(E(j))(y)\times \frac{x-y}{|x-y|^3}d^3y}{4\pi}$. Keeping in mind that $\operatorname{BS}_{\Omega}$ is symmetric with respect to the $L^2(\Omega)$-inner product we find, setting for notational simplicity $X:=\operatorname{BS}_{\partial\Omega}(j)$,
	\begin{gather}
		\nonumber
		\|X\|^2_{L^2(\Omega)}=\langle \operatorname{BS}_{\Omega}(X),\operatorname{curl}(E(j))\rangle_{L^2(\Omega)}+\left\langle X,\frac{\nabla_x\int_{\Omega}E(j)(y)\cdot \frac{x-y}{|x-y|^3}d^3y}{4\pi}-E(j)\right\rangle_{L^2(\Omega)}.
	\end{gather}
	It then follows from the proof of \cite[Lemma 6.2]{G24} that if we let $T(j)(x):=\frac{\nabla_x\int_{\Omega}E(j)(y)\cdot \frac{x-y}{|x-y|^3}d^3y}{4\pi}$, then $\|T(j)\|^2_{L^2(\Omega)}\leq \langle E(j),T(j)\rangle_{L^2(\Omega)}$. This in turn implies that $\|E(j)-T(j)\|_{L^2(\Omega)}\leq \|E(j)\|_{L^2(\Omega)}$ from which we conclude
	\begin{gather}
		\nonumber
		\|X\|^2_{L^2(\Omega)}\leq \|\operatorname{BS}_{\Omega}(X)\|_{L^2(\Omega)}\|\operatorname{curl}(E(j))\|_{L^2(\Omega)}+\|X\|_{L^2(\Omega)}\|E(j)\|_{L^2(\Omega)}
	\end{gather}
	where we used the Cauchy-Schwarz inequality. We can now use the following inequality, see \cite[Equation (1.9)]{FrHe91I}
	\begin{gather}
		\nonumber
		\|\operatorname{BS}_{\Omega}(X)\|_{L^2(\Omega)}\leq \sqrt[3]{\frac{\pi}{16}\operatorname{vol}(\Omega)}\|X\|_{L^2(\Omega)}.
	\end{gather}
	We therefore arrive at the inequality
	\begin{AppB}
		\label{AppBE5}
		\|\operatorname{BS}_{\partial\Omega}(j)\|_{L^2(\Omega)}\leq \sqrt[3]{\frac{\pi}{16}\operatorname{vol}(\Omega)}\|\operatorname{curl}(E(j))\|_{L^2(\Omega)}+\|E(j)\|_{L^2(\Omega)}.
	\end{AppB}
	For given $\epsilon>0$ we set according to \Cref{AppBL2} $\alpha:=\sqrt{\frac{9\epsilon}{9\epsilon+3+\frac{1}{\epsilon}}}$ and $:r=\alpha\rho$ where $\rho$ is the reach of $\partial\Omega$. We then set $c:=\sqrt[3]{\frac{\pi}{16}\operatorname{vol}(\Omega)}$ and square (\ref{AppBE5})
	\begin{gather}
		\nonumber
		\|\operatorname{BS}_{\partial\Omega}(j)\|^2_{L^2(\Omega)}\leq \frac{c^2}{r}\left(r\|\operatorname{curl}(E(j))\|^2_{L^2(\Omega)}\right)+\frac{\|E(j)\|^2_{L^2(\Omega)}}{r}r+2c\|\operatorname{curl}(E(j))\|_{L^2(\Omega)}\|E(j)\|_{L^2(\Omega)}
		\\
		\nonumber
		\leq \left(r\|\operatorname{curl}(E(j))\|^2_{L^2(\Omega)}\right)c^2\left(\frac{1}{r}+\frac{1}{\delta}\right)+\frac{\|E(j)\|^2_{L^2(\Omega)}}{r}(r+\delta)\text{ for all }\delta>0
	\end{gather}
	where we used the elementary inequality $2ab\leq \delta a^2+\frac{b^2}{\delta}$ for all $\delta>0$. We now pick $\delta=\frac{c^2}{r}$ which ensures that $c^2\left(\frac{1}{r}+\frac{1}{\delta}\right)=r+\delta=r+\frac{c^2}{r}$ and gives us
	\begin{gather}
		\nonumber
		\|\operatorname{BS}_{\partial\Omega}(j)\|^2_{L^2(\Omega)}\leq \left(r+\frac{c^2}{r}\right)\left(r\|\operatorname{curl}(E(j))\|^2_{L^2(\Omega)}+\frac{\|E(j)\|^2_{L^2(\Omega)}}{r}\right)=\left(\alpha \rho+\frac{c^2}{\alpha \rho}\right)\|E(j)\|^2_{\epsilon}
	\end{gather}
	with the notation from \Cref{AppBL2}. We now utilise \Cref{AppBL2} and obtain
	\begin{gather}
		\nonumber
		\|\operatorname{BS}_{\partial\Omega}(j)\|^2_{L^2(\Omega)}\leq \sqrt{9\epsilon\left(9\epsilon+3+\frac{1}{\epsilon}\right)}\left(\alpha\rho+\frac{c^2}{\alpha\rho}\right)\|j\|^2_{L^2(\partial\Omega)}
		\\
		\nonumber
		=\left(9\epsilon\rho+\frac{c^2}{\rho}\left(9\epsilon+3+\frac{1}{\epsilon}\right)\right)\|j\|^2_{L^2(\partial\Omega)}
	\end{gather}
	where we used that $\alpha=\sqrt{\frac{9\epsilon}{9\epsilon+3+\frac{1}{\epsilon}}}$. The expression $\left(9\epsilon\rho+\frac{c^2}{\rho}\left(9\epsilon+3+\frac{1}{\epsilon}\right)\right)$ achieves its global minimum at $\epsilon=\frac{1}{3}\frac{1}{\sqrt{\frac{\rho^2}{c^2}+1}}$ and this choice of $\epsilon$ yields
	\begin{gather}
		\nonumber
		\|\operatorname{BS}_{\partial\Omega}(j)\|^2_{L^2(\Omega)}\leq 3\left(\frac{\rho}{\sqrt{\frac{\rho^2}{c^2}+1}}+\frac{c^2}{\rho}\left(1+\frac{1}{\sqrt{\frac{\rho^2}{c^2}+1}}+\sqrt{\frac{\rho^2}{c^2}+1}\right)\right)\|j\|^2_{L^2(\partial\Omega)}
	\end{gather}
		\begin{AppB}
			\label{AppBE6}
		=\frac{3c^2}{\rho}\left(1+2\sqrt{\frac{\rho^2}{c^2}+1}\right)\|j\|^2_{L^2(\partial\Omega)}.
	\end{AppB}
	We notice that the function $s\mapsto 1+2\sqrt{s+1}$ is strictly increasing on $(0,\infty)$ so that we are left with obtaining an upper bound on the quotient $\frac{\rho}{c}$. We recall that $c^3=\frac{\pi}{16}\operatorname{vol}(\Omega)$ and therefore
	\begin{gather}
		\nonumber
		\frac{c^3}{\rho^3}=\frac{\pi}{16}\frac{\operatorname{vol}(\Omega)}{\rho^3}=\frac{\pi^2}{12}\frac{\operatorname{vol}(\Omega)}{\frac{4\pi}{3}\rho^3}.
	\end{gather}
	According to the uniform ball interpretation of the reach, cf. \cite[Section 2]{Dal18}, $\Omega$ must contain a ball of radius $\rho$ so that $\frac{\operatorname{vol}(\Omega)}{\frac{4\pi}{3}\rho^3}\geq 1$ and hence we arrive at $\frac{c^3}{\rho^3}\geq \frac{\pi^2}{12}\Leftrightarrow \frac{\rho}{c}\leq \sqrt[3]{\frac{12}{\pi^2}}$. We can insert this into (\ref{AppBE6}) and find
	\begin{gather}
		\nonumber
		\|\operatorname{BS}_{\partial\Omega}(j)\|_{L^2(\Omega)}\leq \frac{\sqrt[3]{\operatorname{vol}(\Omega)}}{\sqrt{\rho}}\left(\sqrt{3+6\sqrt{\left(\frac{12}{\pi^2}\right)^{\frac{2}{3}}+1}}\right)\sqrt[3]{\frac{\pi}{16}}\|j\|_{L^2(\partial\Omega)}.
	\end{gather}
	This coincides with (\ref{AppBE1}) and thus completes the proof of \Cref{AppBT1}.
\end{proof}
\bibliographystyle{plain}
\bibliography{mybibfileNOHYPERLINK}

\begin{thebibliography}{10}

\bibitem{ABDG98}
C.~Amrouche, C.~Bernardi, M.~Dauge, and V.~Girault.
\newblock Vector {P}otentials in {T}hree-dimensional {N}on-smooth {D}omains.
\newblock {\em Mathematical Methods in the Applied Sciences}, 23:823--864,
  1998.

\bibitem{Dal18}
J.~Dalphin.
\newblock Uniform ball property and existence of optimal shapes for a wide
  class of geometric functionals.
\newblock {\em Interfaces and Free Boundaries}, 20:211--260, 2018.

\bibitem{DK20}
D.W. Dudt and E.~Kolemen.
\newblock {DESC}: {A} stellarator equilibrium solver.
\newblock {\em Phys. Plasmas}, 27:102513, 2020.

\bibitem{Evans10}
L.C. Evans.
\newblock {\em Partial Differential Equations}.
\newblock American Mathematical Society, 2nd edition, 2010.

\bibitem{Fed59}
H.~Federer.
\newblock Curvature {M}easures.
\newblock {\em Transactions of the American Mathematical Society},
  93(3):418--491, 1959.

\bibitem{FrHe91I}
M.H. Freedman and Z.~He.
\newblock Divergence-{F}ree {F}ields: {E}nergy and {A}symptotic {C}rossing
  {N}umber.
\newblock {\em Annals of Mathematics}, 134:189--229, 1991.

\bibitem{G24}
W.~Gerner.
\newblock Properties of the {B}iot-{S}avart operator acting on surface
  currents.
\newblock {\em SIAM J. Math. Anal.}, 56(5):6446--6482, 2024.

\bibitem{G25BiotSavartImageKernelArXiv}
W.~{Gerner}.
\newblock {Kernel and image of the {B}iot-{S}avart operator and their
  applications in stellarator designs}.
\newblock {\em arXiv e-prints}, page
  \href{https://arxiv.org/abs/2504.04176}{arXiv:2504.04176}, April 2025.

\bibitem{G25GaffneyKornArXiv}
W.~{Gerner}.
\newblock {Quantitative Gaffney and Korn inequalities}.
\newblock {\em arXiv e-prints}, page
  \href{https://arxiv.org/abs/2510.05870}{arXiv:2510.05870}, October 2025.

\bibitem{Hel14}
P.~Helander.
\newblock Theory of plasma confinement in non-axisymmetric magnetic fields.
\newblock {\em Rep. Prog. Phys.}, 77:087001, 2014.

\bibitem{HHPH21}
S.A. Henneberg, S.R. Hudson, D.~Pfefferl\'{e}, and P.~Helander.
\newblock Combined plasma-coil optimization algorithms.
\newblock {\em J. Plasma Phys.}, 87:905870226, 2021.

\bibitem{IGPW24}
L.-M. Imbert-Gerard, E.J. Paul, and A.M. Wright.
\newblock {\em An {I}ntroduction to {S}tellarators}.
\newblock SIAM, 2024.

\bibitem{JGLRW23}
R.~Jorge, A.~Goodman, M.~Landreman, J.~Rodrigues, and F.~Wechsung.
\newblock Single-stage stellarator optimization: combining coils with fixed
  boundary equilibria.
\newblock {\em Plasma Phys. Control. Fusion}, 65:074003, 2023.

\bibitem{KLM24}
J.~Kappell, M.~Landreman, and D.~Malhotra.
\newblock The magnetic gradient scale length explains why certain plasmas
  require close external magnetic coils.
\newblock {\em Plasma Phys. Control. Fusion}, 66:025018, 2024.

\bibitem{L17}
M.~Landreman.
\newblock An improved current potential method for fast computation of
  stellarator coil shapes.
\newblock {\em Nucl. Fusion}, 57:046003, 2022.

\bibitem{L18}
J.M. Lee.
\newblock {\em Introduction to Riemannian Manifolds}.
\newblock Springer, second edition, 2018.

\bibitem{M87}
P.~Merkel.
\newblock Solution of stellarator boundary value problems with external
  currents.
\newblock {\em Nucl. Fusion}, 27:867--871, 1987.

\bibitem{M65}
J.W. Milnor.
\newblock {\em Topology from the differentiable view point}.
\newblock Princeton University Press, 1965.

\bibitem{PLBD18}
E.J. Paul, M.~Landreman, A.~Bader, and W.~Dorland.
\newblock An adjoint method for gradient-based optimization of stellarator coil
  shapes.
\newblock {\em Nuclear Fusion}, 58:076015, 2018.

\bibitem{PRS22}
Y.~Privat, R.~Robin, and Sigalotti M.
\newblock Optimal shape of stellarators for magnetic confinement fusion.
\newblock {\em J. Math. Pures Appl.}, 163:231--264, 2022.

\bibitem{Spi58}
L.~Spitzer.
\newblock The {S}tellarator {C}oncept.
\newblock {\em The Physics of Fluids}, 1:253--264, 1958.

\bibitem{Weyl40}
H.~Weyl.
\newblock The method of orthogonal projection in potential theory.
\newblock {\em Duke Math. J.}, 7(1):411--444, 1940.

\bibitem{Xu16}
Y.~Xu.
\newblock A general comparison between tokamak and stellarator plasmas.
\newblock {\em Matter and Radiation at Extremes}, 1(4):192--200, 2016.

\end{thebibliography}
\footnotesize
\end{document}